\documentclass[11pt,a4paper]{article}
\usepackage{a4wide}

\usepackage[utf8]{inputenc}
\usepackage[T1]{fontenc}
\usepackage{amsmath}
\usepackage{amssymb}
\usepackage{url}
\usepackage{graphicx}
\usepackage{here}
\usepackage{xparse}

\usepackage{thmtools}
\usepackage{authblk}
\usepackage{amsthm}

\newtheorem{theorem}{Theorem}
\newtheorem{lemma}{Lemma}
\newtheorem{fact}{Fact}

\newcommand{\Alp}[1][!]{\Sigma_{#1}}
\newcommand{\minpper}{\mathrm{p}}

\NewDocumentCommand{\pper}{o}{\mathrel{\IfNoValueTF{#1}{\parallel}{\parallel_{#1}}}}

\title{The Parameterized Periodicity Lemma}
\date{}

\author[1]{Rikuya~Hamai}
\author[2]{Yuto~Nakashima}
\author[2]{Shunsuke~Inenaga}

\affil[1]{Department of Information Science and Technology, Kyushu University, Japan}
\affil[ ]{\texttt{hamai.rikuya.226@s.kyushu-u.ac.jp}}
\affil[2]{Department of Informatics, Kyushu University, Japan}
\affil[ ]{\texttt{\{nakashima.yuto.003, inenaga.shunsuke.380\}@m.kyushu-u.ac.jp}}

\begin{document}

\maketitle

\begin{abstract}
    Fine and Wilf [Proc. Amer. Math. Soc. 1965] showed that any string of length at least $p+q-d$ with periods $p$ and $q$ also has period $d=\gcd(p,q)$. For parameterized strings, Apostolico and Giancarlo [Discrete Appl. Math. 2008] proved an analogue with length bound $p+q$, assuming that the two induced bijections commute. Ideguchi et al. [SPIRE 2023] removed this assumption and gave the bound $p+q+\min(p,q)(\sigma-1)$, where $\sigma$ is the number of distinct letters.
    This was later improved by Hamai et al. [SPIRE 2024] to $p+q+\min(p,q)(\sigma-2)$, which was used to bound the number of non-equivalent parameterized squares.
    In this paper, we establish the optimal Fine--Wilf type bound for parameterized strings.
    Namely, if a string $s$ containing $\sigma$ distinct letters has parameterized periods $p$ and $q$ and satisfies
    $|s| \ge p+q+(\sigma-3)d+1$, where $d=\gcd(p,q)$,
    then $d$ is also a parameterized period of $s$.
    We also give matching lower-bound instances, proving that our bound is optimal for any $\sigma \geq 2$.
\end{abstract}

\section{Introduction}

Periodicity is one of the central notions in combinatorics on words and string algorithms. For a string $s$ of length $n$, an integer $p$ is a period of $s$ if $s[i]=s[i+p]$ holds for every $0\le i<n-p$. A classical theorem of Fine and Wilf states that if a string has periods $p$ and $q$ and its length is at least $p+q-\gcd(p,q)$, then $\gcd(p,q)$ is also a period~\cite{Fine1965}. This theorem shows that when two period constraints coexist on a sufficiently long string, they force a finer common period. Periodicity is also a fundamental notion behind repetitions and runs, classical objects in string algorithms~\cite{Bannai2017,CrochemoreIlie2008,Crochemore1981,Kolpakov1999,MainLorentz1984}.

Parameterized matching, introduced by Baker~\cite{Baker1993,Baker1996}, is a classical extension of exact string matching in which strings are compared up to a consistent renaming of alphabet symbols. For two strings $x,y$ of the same length and a bijection $f$ on the alphabet, we say that $x$ and $y$ parameterized-match under $f$ if $y[i]=f(x[i])$ holds for every position $i$. This model captures structural equality rather than literal equality, and was motivated in part by duplicated-code detection and software maintenance~\cite{Baker1997}. Parameterized matching has since been studied in several algorithmic settings, including alphabet-dependent algorithms, multiple-pattern matching, parameterized indexing data structures, compact indexes, and streaming algorithms~\cite{AmirFM94,FujisatoNIBT19,IduryS96,IseriIHKYS24,JalseniusPS13,NakashimaFHNYIB22}.

Using this notion, parameterized periodicity is defined as a direct analogue of ordinary periodicity. Let $s$ be a string of length $n$. A shift $p$ is a parameterized period of $s$ if there exists a bijection $f$ on the alphabet such that $s[i+p]=f(s[i])$ holds for every $0\le i<n-p$. Thus, unlike ordinary periods, parameterized periods are accompanied by bijections on the alphabet. This additional algebraic structure makes the interaction between two parameterized periods substantially more delicate, since two periods may be realized by different, possibly non-commuting bijections.

The study of periodicity and repetitions in parameterized strings was initiated by Apostolico and Giancarlo~\cite{Apostolico2008}. They showed that if $p$ and $q$ are parameterized periods of $s$ realized by bijections $f$ and $g$, respectively, and if $f$ and $g$ commute, then $d=\gcd(p,q)$ is also a parameterized period of $s$ provided that $|s|\ge p+q$. Later, Ideguchi et al. gave a parameterized periodicity lemma that does not assume commutativity of the bijections~\cite{Ideguchi2023}. Their result implies that, when $s$ contains $\sigma$ distinct letters, the condition
$|s|\ge p+q+\min(p,q)(\sigma-1)$
is sufficient to guarantee that $d$ is a parameterized period. Hamai et al. further improved this bound to
$|s|\ge p+q+\min(p,q)(\sigma-2)$
and used it in their analysis of non-equivalent parameterized squares~\cite{Hamai2024}. Periodicity lemmas naturally enter this analysis, since overlaps of parameterized squares lead to situations where several parameterized periods coexist.

The remaining dependence on $\min(p,q)$ is not merely a quantitative issue. The value $d=\gcd(p,q)$ captures the common arithmetic structure forced by the two shifts, whereas $\min(p,q)$ can be arbitrarily larger than $d$. Thus, replacing $\min(p,q)$ by $d$ recovers the Fine--Wilf type dependence on the true common period scale in the parameterized setting.

In this paper, we prove an optimal parameterized periodicity lemma. Let $s$ be a string containing $\sigma$ distinct letters, and let $p$ and $q$ be parameterized periods of $s$. Let $d=\gcd(p,q)$. We show that $d$ is also a parameterized period of $s$ whenever
$|s| \ge p+q+(\sigma-3)d+1$.
This gives the Fine--Wilf type dependence on the gcd scale. Moreover, the bound is tight: for every $\sigma\ge 2$, we construct strings of length
$p+q+(\sigma-3)d$
having parameterized periods $p$ and $q$ but not parameterized period $d$.

\section{Preliminaries}

\subsection{Strings}
Let $\Sigma$ be an {\em alphabet}.
An element of $\Sigma^*$ is called a {\em string}.
The length of a string $s$ is denoted by $|s|$.
The empty string $\varepsilon$ is the string of length 0.
Let $\Sigma^+$ be the set of non-empty strings,
i.e., $\Sigma^+ = \Sigma^* \setminus \{\varepsilon \}$.
For any strings $x$ and $y$,
let $x \cdot y$ (or sometimes $xy$) denote the concatenation of the two strings.
For a string $s = xyz$, $x$, $y$ and $z$ are called
a \emph{prefix}, \emph{substring}, and \emph{suffix} of $s$, respectively.
The $i$-th symbol of a string $w$ is denoted by $w[i]$, where $0 \leq i < |w|$.
For a string $w$ and two integers $0 \leq i \leq j < |w|$,
let $w[i..j]$ denote the substring of $w$ that begins at position $i$ and ends at
position $j$. For convenience, let $w[i..j] = \varepsilon$ when $i > j$.
Also,
let
$w[..i]=w[0..i]$,
$w[i..]=w[i..|w|-1]$, and
$w[i..j] = w(i-1..j] = w[i..j+1)$.
For any string $w$, let $w^1 = w$ and let $w^k = ww^{k-1}$ for any integer $k \ge 2$.
For a string $w$, let $\Sigma_w$ denote the set of distinct letters
occurring in $w$.

\subsection{Periods}
Let $s$ be a string of length $n$.
An integer $p$ with $1 \le p \le n$ is called a period of $s$ if
$s[i]=s[i+p]$ for every $0 \le i < n-p$.
Equivalently,
$s[0..n-p) = s[p..n)$.

\begin{lemma}[Theorem~1 of~\cite{Fine1965}]\label{lem:periodicity}
    Let $s$ be a string such that $p$ and $q$ are periods of $s$.
    If $|s| \ge p+q-\gcd(p,q)$, then $\gcd(p,q)$ is a period of $s$.
\end{lemma}

\subsection{Parameterized equivalence}
Two strings $x$ and $y$ of length $k$ each are said to be \emph{parameterized equivalent} 
iff there is a bijection $f$ on $\Sigma$ such that $f(x[i]) = y[i]$ for all $0 \leq i < k$.
For instance,
let $\Sigma = \{\mathtt{a}, \mathtt{b}, \mathtt{c}, \mathtt{d}\}$,
and consider two strings $x = \mathtt{aabcacbbdad}$ and $y = \mathtt{bbcabaccdbd}$.
These two strings are parameterized equivalent, since $x$ can be transformed to $y$ by applying a bijection $f$
such that $f(\mathtt{a}) = \mathtt{b}$, $f(\mathtt{b}) = \mathtt{c}$,
$f(\mathtt{c}) = \mathtt{a}$, and $f(\mathtt{d}) = \mathtt{d}$ to the characters in $x$.
We write $x \approx y$ iff $x$ and $y$ are parameterized equivalent.

\subsection{Parameterized periods}
Let $s$ be a string of length $n$.
An integer $p$ with $1 \le p \le n$ is said to be a parameterized period,
or a p-period, of $s$ if
$s[0..n-p) \approx s[p..n)$.
For any p-period $p$ of $s$, we write $p \pper[f] s$ if
$s[0..n-p) \approx s[p..n)$ holds by a bijection $f$.
We sometimes drop the subscript $f$ and write $p \pper s$ when no confusion occurs.
The smallest p-period of $s$ is denoted by $\minpper(s)$.
By the definition of p-periods, we obtain the following facts which we will use in our proof.

\begin{fact}\label{fact:restrict-alphabet}
Let $s$ be a string and let $p$ be a p-period of $s$. Then an associated
bijection for $p$ can be chosen as a permutation of $\Sigma_s$. 
When the entire alphabet $\Sigma$ is under consideration, the values on $\Sigma \setminus \Sigma_s$ are
irrelevant to the period condition and can be chosen arbitrarily.
\end{fact}

\begin{fact}[cf.~\cite{SCER-period}]
    Let $s$ be a string that satisfies $p \pper[f] s$.
    For any position $i$ satisfying $0 \leq i < |s|-p$, $f(s[i]) = s[i+p]$.
\end{fact}

Apostolico and Giancarlo~\cite{Apostolico2008} showed a variant of the periodicity lemma on the parameterized equivalence model.

\begin{lemma}[Lemma~3 of~\cite{Apostolico2008}] \label{lem:prev-p-perlem1}
    Let $s$ be a string that satisfies $p \pper[f] s$ and $q \pper[g] s$.
    If $p+q \leq |s|$ and $f \circ g = g \circ f$, then $\gcd(p, q) \pper s$.
\end{lemma}

The above lemma uses the commutativity of the two bijections.
Hamai et al.~\cite{Hamai2024} proposed a lemma regarding the commutativity of the two bijections.

\begin{lemma}[Lemma~5 of~\cite{Hamai2024}] \label{lem:permutation}
    Assume that $|\Sigma| \geq 2$.
    Let $\Sigma'$ be a subset of $\Sigma$ satisfying $|\Sigma'| = |\Sigma|-2$.
    For any permutations $f$ and $g$ of $\Sigma$,
    if $(f \circ g) (a) = (g \circ f) (a)$ for all $a \in \Sigma'$,
    then $f$ and $g$ commute.
\end{lemma}

Recently, Ideguchi et al.~\cite{Ideguchi2023} proposed the following variant of the parameterized periodicity lemma which does not assume commutativity of the bijections:

\begin{lemma}[Lemma~5 of~\cite{Ideguchi2023}] \label{lem:prev-p-perlem2}
    Let $s$ be a string that satisfies $p \pper s$ and $q \pper s$.
    If $p+q +\min(p,q) \cdot (|\Alp[s]|-1) \leq |s|$, then $\gcd(p, q) \pper s$.
\end{lemma}

Moreover, Hamai et al.~\cite{Hamai2024} improved Lemma~\ref{lem:prev-p-perlem2}.

\begin{lemma}[Lemma~6 of~\cite{Hamai2024}]\label{lem:prev-p-perlem3}
  Let $s$ be a string that satisfies $p \pper s$ and $q \pper s$.
  If $p + q + \min(p, q) \cdot (|\Alp[s]| - 2) \leq |s|$, then $\gcd(p, q) \pper s$ holds.
\end{lemma}

The preceding bounds indicate that the number of distinct letters in the string is a key
parameter in parameterized periodicity. The following lemma of Hamai et al.~\cite{Hamai2024}
will be used to relate parameterized periods to the number of distinct letters in substrings.

\begin{lemma}[Lemma~4 of~\cite{Hamai2024}]\label{lem:substr-character-occ}
    Let $s$ be a string that satisfies $p \pper s$.
    For any substring $s'$ of $s$ and any integer $k$ satisfying $1 \leq k \leq |\Alp[s]| $, 
    if $|s'| \geq p \cdot (k-1) + 1$,
    then $|\Alp[s']| \geq k$ holds.
\end{lemma}

The known parameterized periodicity lemmas above suggest that the dependence on the
number of distinct letters is essential. In particular, the best previous bound contains
the additive term $\min(p,q)(|\Sigma_s|-2)$. In contrast, we prove that this dependence
is governed by $d=\gcd(p,q)$ rather than by $\min(p,q)$. \section{The Parameterized Periodicity Lemma}
Throughout this section, we fix a string $s$ and let $\sigma = |\Sigma_s|$.
We now state our main result.

\begin{theorem}\label{thm:periodicity}
Let $s$ be a string such that $p \pper s$ and $q \pper s$. If $|s| \ge p+q+d(\sigma-3)+1$, then $d \pper s$, where $d=\gcd(p,q)$.
\end{theorem}

If $\sigma=1$, then $s$ consists of a single letter, and hence every positive
integer $r \le |s|$ is a parameterized period of $s$. Thus there is nothing to
prove in this case. In what follows, we assume $\sigma \ge 2$.

The proof is divided into two parts. We first prove the binary case, using the special
structure of bijections on a two-letter alphabet. We then prove the general case by deriving
a local $d$-step transition from the two parameterized periods and combining it with an
alphabet-counting argument.

 \subsection{The Parameterized Periodicity Lemma for $\sigma =2$}

We first prove Theorem~\ref{thm:periodicity} for $\sigma=2$.
By renaming the letters occurring in $s$, we may assume that
$\Sigma_s=\{0,1\}$. By Fact~\ref{fact:restrict-alphabet}, the associated
bijections for the parameterized periods can be regarded as permutations of
$\Sigma_s$.
Every bijection on $\{0,1\}$ is either the identity or the transposition.
Hence each such bijection is represented by a bit $\gamma\in\{0,1\}$:
\[
f(a)=a\oplus\gamma
\]
for every $a\in\{0,1\}$, where $\oplus$ denotes addition modulo two.

We will use the elementary identities
\[
((u+v)\bmod 2)\kappa
=
(u\bmod 2)\kappa\oplus(v\bmod 2)\kappa
\]
and
\[
(u\bmod 2)\oplus((u+1)\bmod 2)=1
\]
for all integers $u,v$ and all $\kappa\in\{0,1\}$.

\begin{proof}[Proof of Theorem~\ref{thm:periodicity} for $\sigma=2$]
Let $n=|s|$.
Since $p\pper s$ and $q\pper s$, there exist bits
$\alpha,\beta\in\{0,1\}$ such that $s[i+p]=s[i]\oplus\alpha$
for every $0\le i<n-p$, and
$s[i+q]=s[i]\oplus\beta$
for every $0\le i<n-q$.

Let $d=\gcd(p,q)$, and write $p=ad$ and $q=bd$ with $\gcd(a,b)=1$.
Since $\sigma=2$, the length condition gives
\[
n\ge p+q-d+1=(a+b-1)d+1.
\]
It remains to show that there exists $\kappa\in\{0,1\}$ such that
$
s[i+d]=s[i]\oplus\kappa
$
for every $0\le i<n-d$.

For each residue $r\in\{0,\ldots,d-1\}$, define
$
s_r[t]=s[r+td]
$
for all $t$ with $0\le r+td<n$, and let $m_r=|s_r|$.
Then
$
s_r[t+a]=s_r[t]\oplus\alpha
$
for every $0\le t<m_r-a$, and
$
s_r[t+b]=s_r[t]\oplus\beta
$
for every $0\le t<m_r-b$.
Moreover,
\[
m_r=\left\lfloor\frac{n-1-r}{d}\right\rfloor+1
\ge a+b-1.
\]
For $r=0$, we have $m_0\ge a+b$.

Since $s_0$ has parameterized periods $a$ and $b$, and since the two binary bijections
commute, Lemma~\ref{lem:prev-p-perlem1} implies that $1\pper s_0$.
Thus there exists $\kappa\in\{0,1\}$ such that
\[
s_0[t+1]=s_0[t]\oplus\kappa
\]
for every $0\le t<m_0-1$.

By applying this identity repeatedly, for every $0\le t<m_0-a$ we have
\[
  s_0[t+a]
  = s_0[t]\oplus
    \underbrace{\kappa\oplus\cdots\oplus\kappa}_{a}
  = s_0[t]\oplus (a\bmod 2)\kappa .
\]
On the other hand, since $s_0[t+a]=s_0[t]\oplus\alpha$ holds for every
$0\le t<m_0-a$, we obtain $\alpha=(a\bmod 2)\kappa$. Similarly,
$\beta=(b\bmod 2)\kappa$.

Now fix any residue $r$ and define
$
v_r[t]=s_r[t]\oplus(t\bmod 2)\kappa.
$
We show that $v_r$ has ordinary periods $a$ and $b$. For $0\le t<m_r-a$,
\[
\begin{aligned}
v_r[t+a]
&=s_r[t+a]\oplus((t+a)\bmod 2)\kappa\\
&=s_r[t]\oplus\alpha\oplus(t\bmod 2)\kappa\oplus(a\bmod 2)\kappa\\
&=v_r[t].
\end{aligned}
\]
Thus $a$ is a period of $v_r$. The same argument, using
$\beta=(b\bmod 2)\kappa$, shows that $b$ is also a period of $v_r$.

Since $m_r\ge a+b-1$ and $\gcd(a,b)=1$, Lemma~\ref{lem:periodicity} implies that $1$ is a
period of $v_r$. Hence
\[
v_r[t+1]=v_r[t]
\]
for every $0\le t<m_r-1$. Therefore,
\[
\begin{aligned}
s_r[t+1]
&=s_r[t]\oplus(t\bmod 2)\kappa\oplus((t+1)\bmod 2)\kappa\\
&=s_r[t]\oplus\kappa.
\end{aligned}
\]
Since this holds for every residue $r$, we obtain
\[
s[i+d]=s[i]\oplus\kappa
\]
for every $0\le i<n-d$. Therefore, $d\pper s$.
\end{proof}
 \subsection{The Parameterized Periodicity Lemma for $\sigma \ge3$}
To prove Theorem~\ref{thm:periodicity} for $\sigma\ge3$, we first establish the following lemma.

\begin{lemma}\label{lem:period-d}
Let $f$ and $g$ be permutations of an alphabet $\Sigma$, and let
$p\le q$ be positive integers. Put $d=\gcd(p,q)$.
For each $i$ with $0\le i<p-d$, there exists a permutation $h_i$,
depending only on $f,g,p,q$, and $i$, such that
\[
s[i+d]=h_i(s[i])
\]
for every string $s\in\Sigma^*$ satisfying
$p\pper[f]s$, $q\pper[g]s$, and $|s|\ge p+q$.
\end{lemma}

\begin{proof}
If $q$ is a multiple of $p$, then $d=\gcd(p,q)=p$ and the index range $0\le i < p-d$ is empty, so the statement is trivial.
Hence we assume that $q$ is not a multiple of $p$.

Define $a_0=i$. For each $t\ge0$, let $k_t\in\mathbb{Z}_{\ge0}$ and
$a_{t+1}\in\{0,\ldots,p-1\}$ be uniquely determined by
\[
a_t+q=k_t p+a_{t+1}.
\]
Then $a_{t+1}\equiv a_t+q \pmod p$ for all $t\ge 0$.
By induction on $t$, we obtain
\[
  a_t\equiv a_0+tq \pmod p
\]
for every $t\ge 0$.

We claim that there exists $T\ge 0$ such that $a_T=a_0+d$.
Indeed, writing $p=dn$ and $q=dm$ with $\gcd(n,m)=1$, the congruence $a_T\equiv a_0+d\pmod p$ is equivalent to
\[
  a_0+Tq\equiv a_0+d \pmod p
  \Longleftrightarrow
  Tq\equiv d \pmod p
  \Longleftrightarrow
  Tm\equiv 1 \pmod n .
\]
Since $\gcd(n,m)=1$, such a $T$ exists.
Let $T_{\min}$ be the smallest $T\ge 0$ such that $a_T\equiv a_0+d\pmod p$.
Since $0\le a_{T_{\min}}<p$ and $0\le a_0+d<p$, we have $a_{T_{\min}}=a_0+d$.

Next, for each $t\ge 0$, we relate $s[a_{t+1}]$ and $s[a_t]$.
From $q \pper[g] s$, we have
\[
  s[a_t+q]=g(s[a_t]).
\]
This application of the $q$-period is valid because
\[
a_t<p\le |s|-q.
\]
Also, all applications of the $p$-period are valid. Indeed, for every $0 \le \ell < k_t$, we have
\[
a_{t+1}+(\ell+1)p \le a_{t+1}+k_t p = a_t+q < p+q \le |s|.
\]
Hence
\[
a_{t+1}+\ell p < |s|-p.
\]
Therefore, by applying the $p$-period $k_t$ times, we obtain
\[
s[a_t+q]=s[a_{t+1}+k_t p]=f^{k_t}(s[a_{t+1}]).
\]
Hence,
\[
s[a_{t+1}]=f^{-k_t}\circ g(s[a_t]).
\]
By iterating this equation, we obtain
\[
  s[a_0+d]=s[a_{T_{\min}}]
  =
  (f^{-k_{T_{\min}-1}}\circ g)
  \circ \cdots \circ
  (f^{-k_0}\circ g)(s[a_0]).
\]
Define
\[
  h_i :=
  (f^{-k_{T_{\min}-1}}\circ g)
  \circ \cdots \circ
  (f^{-k_0}\circ g).
\]
Then $s[i+d]=h_i(s[i])$.
Moreover, the construction of $h_i$ depends only on $f,g,p,q$, and $i$.
\end{proof}

With the above lemma established, we are now ready to prove Theorem~\ref{thm:periodicity}.

\begin{proof}[Proof of Theorem~\ref{thm:periodicity} for $\sigma\ge3$]
By Fact~\ref{fact:restrict-alphabet}, choose permutations $f$ and $g$ of
$\Sigma_s$ such that $p \pper[f] s$ and $q \pper[g] s$.
By symmetry, assume $p\le q$.
If $q$ is a multiple of $p$, then $d=p$, and hence $d \pper s$ is immediate.
Thus, assume that $q$ is not a multiple of $p$.

The remainder of the proof consists of three steps.
We first bound the alphabet of a short prefix of $s$, then propagate this
alphabet, and finally show that it coincides with the alphabet of the entire
string.

\smallskip
\noindent\textbf{Step~1 (Bounding $|\Sigma_{s_0}|$).}
We first bound the number of distinct letters occurring in the prefix $s_0$.
Since $\sigma\ge3$, the length condition gives
\[
|s| > p+q.
\]
Therefore, by Lemma~\ref{lem:prev-p-perlem1}, it suffices to prove that the bijections
$f$ and $g$ commute.

Suppose, for a contradiction, that $f$ and $g$ do not commute.

Put $x=d(\sigma-3)+1$ and $s_0=s[0..x).$
Since $|s|\ge p+q+x$, for every $a\in\Sigma_{s_0}$ we have
\[
(f\circ g)(a)=(g\circ f)(a).
\]
Indeed, $(f\circ g)(s[i])=s[i+q+p]$ and $(g\circ f)(s[i])=s[i+p+q]$
for every $0\le i<x$.
Since $f$ and $g$ do not commute, the contrapositive of Lemma~\ref{lem:permutation}, applied to the alphabet
$\Sigma_s$, gives
\[
|\Sigma_{s_0}| \leq |\Sigma_s|-3 = \sigma-3.
\]
If $\sigma=3$, then $x=1$, and hence $s_0$ is non-empty.
This contradicts $|\Sigma_{s_0}|\le 0$.

Thus, in the rest of the proof, we assume $\sigma\ge4$.

\smallskip
\noindent\textbf{Step~2 (Propagating $\Sigma_{s_0}$).}
We next show that the alphabet of $s_0$ propagates to a longer prefix.
Since $|s|\ge p+q+x$, for every $b$ with $0\le b<x$, the suffix $s[b..]$
has length at least $p+q$. Moreover, each such suffix inherits the same
$p$-period and $q$-period with the same bijections $f$ and $g$. Hence
Lemma~\ref{lem:period-d} can be applied to $s[b..]$ for every
$0\le b<x$.
For each integer $i$ with $0\le i\le p/d-2$, put
\[
r_i=i \cdot d,\qquad
s_i=s[r_i..r_i+x),\qquad
y_i=s[r_i..r_i+x+d).
\]
The relative positions of these intervals are illustrated in Figure~\ref{fig:s_iy_i}.
\begin{figure}[t]
  \centering
  \includegraphics[keepaspectratio,width=\linewidth]{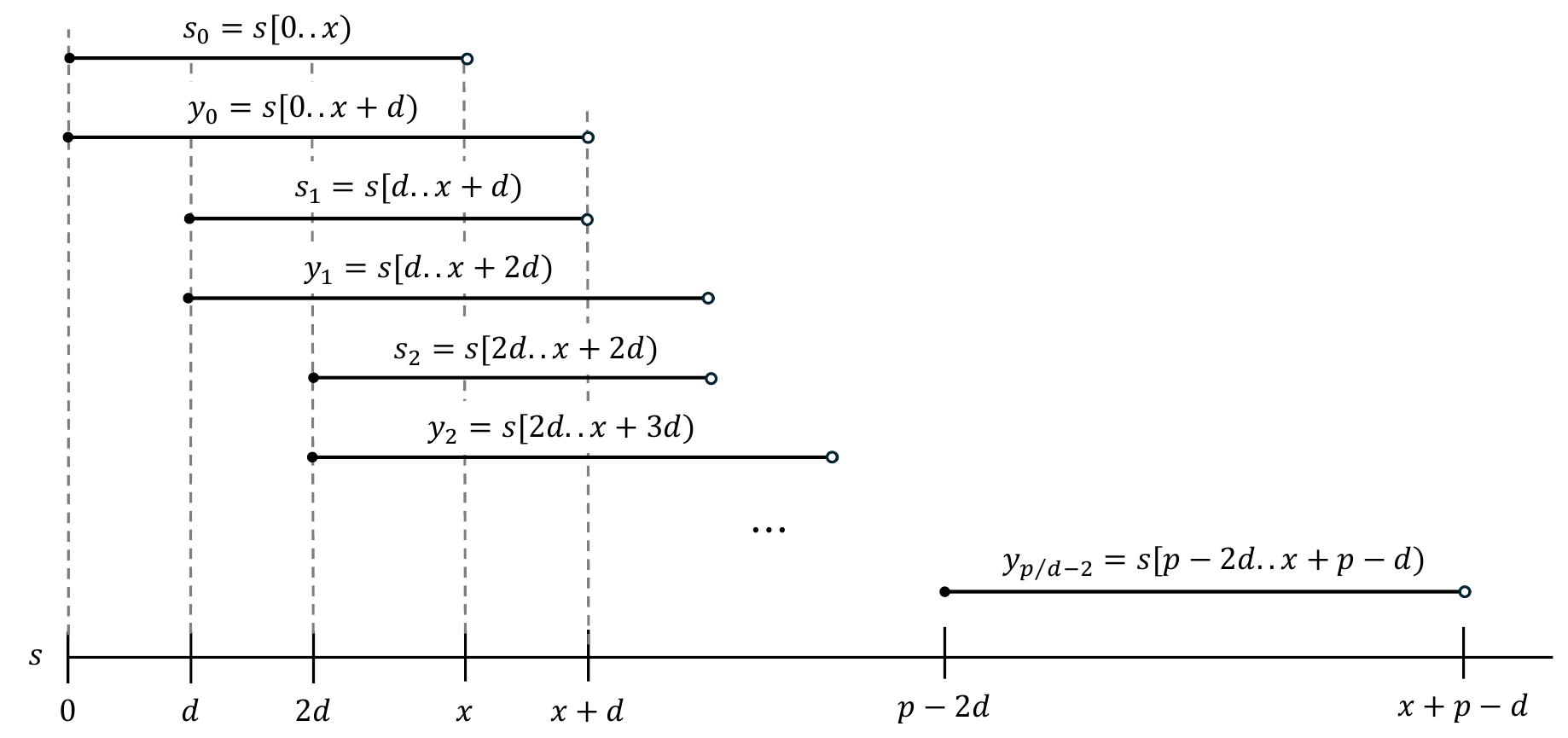}
  \caption{The intervals $s_i=s[id..id+x)$ and $y_i=s[id..id+x+d)$ used in Section~3.2. Each $y_i$ extends $s_i$ by $d$ positions, and the intervals are shifted by $d$ positions as $i$ increases. }
  \label{fig:s_iy_i}
\end{figure}

We first show that $d \pper[h_{r_i}] y_i .$
Let $j$ be any position with $r_i\le j<r_i+x$, and put $b=j-r_i$.
Then $0\le b<x$.
Applying Lemma~\ref{lem:period-d} to the suffix $s[b..]$ and considering relative position $r_i$ (See Figure~\ref{fig:suffix}), we obtain
\[
  s[j+d]
  =
  s[b+r_i+d]
  =
  h_{r_i}(s[b+r_i])
  =
  h_{r_i}(s[j]).
\]
Here the same bijection $h_{r_i}$ is obtained for all choices of $b$, because $h_{r_i}$ depends only on $f,g,p,q$, and $r_i$.
Therefore, $d \pper[h_{r_i}] y_i .$

\begin{figure}[t]
  \centering
  \includegraphics[keepaspectratio,width=\linewidth]{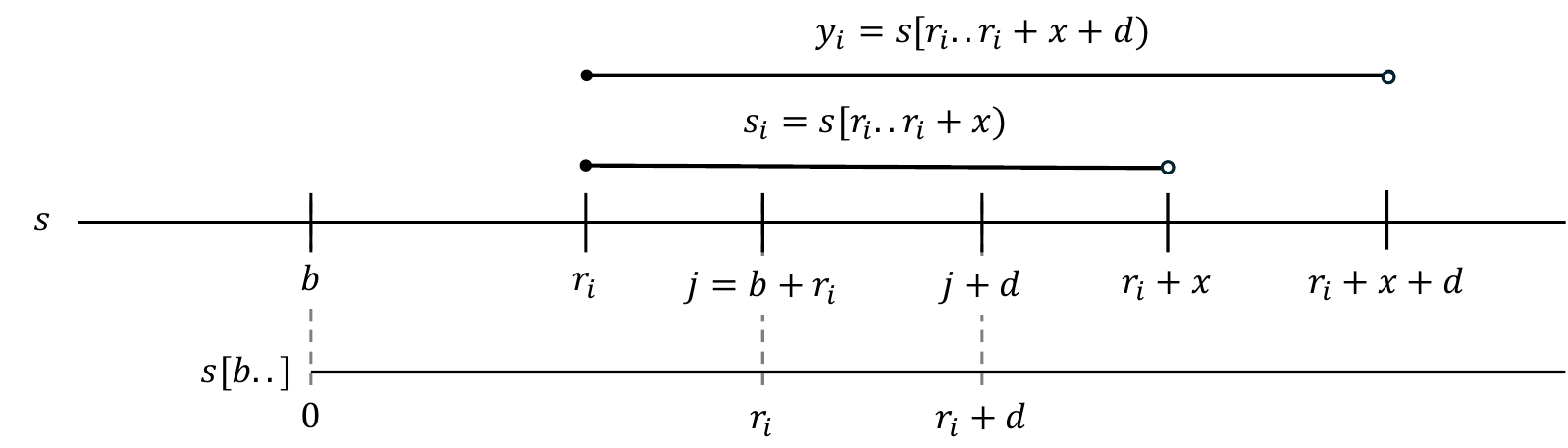}
  \caption{Applying Lemma~\ref{lem:period-d} to suffixes. For a position $j\in[r_i,r_i+x)$, we put $b=j-r_i$. Then the position $j=b+r_i$ in $s$ corresponds to relative position $r_i$ in the suffix $s[b..]$. Hence Lemma~\ref{lem:period-d} gives $s[j+d]=h_{r_i}(s[j])$, which proves $d\pper[h_{r_i}] y_i$.
}
  \label{fig:suffix}
\end{figure}
Since $d \pper[h_{r_i}] y_i$, we have
\[
s_{i+1}=h_{r_i}(s_i)
  \qquad
  (0\le i\le p/d-3).
\]
Hence
\[
  |\Alp[s_{i+1}]|=|\Alp[s_i]|
  \qquad
  (0\le i\le p/d-3),
\]
and therefore
\[
  |\Alp[s_i]|=|\Alp[s_0]|
  \qquad
  (0\le i\le p/d-2).
\]
In particular, $|\Alp[s_i]|\le \sigma-3$.

We next show that
\[
  |\Alp[y_i]|=|\Alp[s_i]|
  \qquad
  (0\le i\le p/d-2)
\]
by using Lemma~\ref{lem:substr-character-occ}.
Since $s_i$ is a substring of $y_i$, suppose to the contrary that
\[
  |\Alp[y_i]|\ge |\Alp[s_i]|+1.
\]
Then Lemma~\ref{lem:substr-character-occ} is applicable to $y_i$ and $s_i$ with $k=|\Alp[s_i]|+1.$
Indeed, $1\le k\le |\Alp[y_i]|$, and
\[
  |s_i|
  =
  x
  =
  d(\sigma-3)+1
  \ge
  d|\Alp[s_i]|+1
  =
  d(k-1)+1.
\]
Thus Lemma~\ref{lem:substr-character-occ} gives
\[
  |\Alp[s_i]|\ge k=|\Alp[s_i]|+1,
\]
a contradiction.
Therefore,
\[
  |\Alp[y_i]|=|\Alp[s_i]|=|\Alp[s_0]|
  \qquad
  (0\le i\le p/d-2).
\]
Since $s_0$ is a substring of $y_0$ and $|\Alp[y_0]|=|\Alp[s_0]|$, we have
\[
  \Alp[y_0]=\Alp[s_0].
\]
Now we show that $\Alp[y_i]=\Alp[y_{i+1}]$ for every $0\le i\le p/d-3$.
The two substrings
\[
  y_i=s[r_i..r_i+x+d),
  \qquad
  y_{i+1}=s[r_i+d..r_i+x+d+d)
\]
overlap exactly on
\[
  s_{i+1}=s[r_i+d..r_i+x+d).
\]
Hence
\[
  \Alp[s_{i+1}]
  \subseteq
  \Alp[y_i]\cap\Alp[y_{i+1}].
\]
Since
\[
  |\Alp[y_i]|
  =
  |\Alp[y_{i+1}]|
  =
  |\Alp[s_{i+1}]|
  =
  |\Alp[s_0]|,
\]
we obtain $\Alp[y_i]=\Alp[y_{i+1}]$.
Therefore,
\[
  \Alp[y_0]=\Alp[y_1]=\cdots=\Alp[y_{p/d-2}]=\Alp[s_0].
\]
Moreover, since $r_i=id$, the substrings $y_i$ cover the interval
\[
[0,x+p-d).
\]
Indeed, $y_0$ starts at position $0$, and the last substring
$y_{p/d-2}$ ends at position
\[
r_{p/d-2}+x+d
=
(p/d-2)d+x+d
=
x+p-d.
\]
Consecutive substrings overlap because $y_i$ ends at $r_i+x+d$ whereas
$y_{i+1}$ starts at $r_i+d$, and $x\ge1$.
Therefore,
\[
\Sigma_{s[0..x+p-d)}=\Sigma_{s_0}.
\]

\smallskip
\noindent\textbf{Step~3 (Showing $\Sigma_s=\Sigma_{s_0}$).}
Finally, we show that $\Sigma_s=\Sigma_{s_0}$.
To this end, we first prove that $f$ preserves $\Sigma_{s_0}$.
Put
\[
  z=s[0..x-d).
\]
Since $z$ is a substring of $y_0$, $d \pper[h_0] y_0$, and
\[
  |z|
  =
  x-d
  =
  d(\sigma-4)+1
  \ge
  d(|\Alp[s_0]|-1)+1,
\]
Lemma~\ref{lem:substr-character-occ} applied to $y_0$ and $z$ with $k=|\Alp[s_0]|$ implies
\[
  |\Alp[z]|\ge |\Alp[s_0]|.
\]
Since $z$ is a substring of $s_0$, we have
\[
  \Alp[z]=\Alp[s_0].
\]
Since $p \pper[f] s$, we have
\[
  s[p..p+x-d)=f(z).
\]
The substring $s[p..p+x-d)$ is contained in $s[0..x+p-d)$, and hence
\[
  f(\Alp[s_0])\subseteq \Alp[s_0].
\]

We finally show that every letter of $s$ belongs to $\Sigma_{s_0}$.
We already know that
\[
\Sigma_{s[0..x+p-d)}=\Sigma_{s_0}.
\]
We prove by induction on $t$ that
\[
s[t]\in\Sigma_{s_0}
\]
for every position $t\ge x+p-d$.

Let $t\ge x+p-d$, and assume that the claim holds for all positions
smaller than $t$. Since $p\pper[f]s$, we have
\[
s[t]=f(s[t-p]).
\]
If $t-p<x+p-d$, then
\[
s[t-p]\in\Sigma_{s_0}
\]
by the prefix condition. Otherwise, since $t-p<t$, the induction hypothesis gives
\[
s[t-p]\in\Sigma_{s_0}.
\]
In both cases, we have $s[t-p]\in\Sigma_{s_0}$. Since
\[
f(\Sigma_{s_0})\subseteq\Sigma_{s_0},
\]
it follows that
\[
s[t]\in\Sigma_{s_0}.
\]
Therefore,
\[
\Sigma_s=\Sigma_{s_0}.
\]
This contradicts $|\Sigma_s|=\sigma$ and $|\Sigma_{s_0}|\le\sigma-3$.
Hence $f$ and $g$ commute. 
By Lemma~\ref{lem:prev-p-perlem1}, we obtain $d \pper s$.
\end{proof}

 \section{Tightness of the bound}
We show that the length bound in Theorem~\ref{thm:periodicity} is tight.
The case $\sigma=1$ is degenerate, since every positive shift at most $|s|$
is a parameterized period. Hence the tightness question is meaningful only for $\sigma \ge 2$.
For every $\sigma\ge2$, we construct a string $s$ with exactly $\sigma$ distinct letters
and integers $p,q$ such that
\[
p\pper s,\qquad q\pper s,\qquad |s|=p+q+d(\sigma-3),
\]
where $d=\gcd(p,q)$, but $d\not\pper s$.
Thus the additive constant $+1$ in Theorem~\ref{thm:periodicity} cannot be removed.

\begin{theorem}For any $\sigma \ge 2$, there exist matching lower-bound instances for
Theorem~\ref{thm:periodicity}.
\end{theorem}

\begin{proof}
Let $n$ be any positive integer, and put $p=2n$, $q=3n$, and
$d=\gcd(p,q)=n$.

First, consider the case $\sigma=2$. Let
\[
s=a_1^n a_2^n a_2^n a_1^n.
\]
Then $|\Sigma_s|=2$ and $|s|=4n$. Moreover, $s$ satisfies
$2n \pper[f] s$ and $3n \pper[g] s$, where
\[
f(a_1)=a_2,\quad f(a_2)=a_1,
\]
and
\[
g(a_1)=a_1,\quad g(a_2)=a_2.
\]
However, $n$ is not a p-period of $s$. Indeed, if $n \pper[h] s$ for some
bijection $h$, then the first block $a_1^n$ and the second block $a_2^n$
force $h(a_1)=a_2$. The second block $a_2^n$ and the third block $a_2^n$
force $h(a_2)=a_2$. The third block $a_2^n$ and the fourth block $a_1^n$
force $h(a_2)=a_1$, a contradiction.

Finally,
\[
|s|=4n=2n+3n+(\sigma-3)n=p+q+(\sigma-3)d.
\]

Next, consider the case $\sigma \ge 3$. Let
\[
s=a_1^n a_2^n \cdots a_\sigma^n a_2^n a_1^n.
\]
Then $|\Sigma_s|=\sigma$ and $|s|=(\sigma+2)n$. Moreover, $s$ satisfies
$2n \pper[f] s$ and $3n \pper[g] s$, where
\[
f(a_i)=a_{i+2}\quad (1 \le i \le \sigma-2),\quad
f(a_{\sigma-1})=a_2,\quad f(a_\sigma)=a_1,
\]
and
\[
g(a_i)=a_{i+3}\quad (1 \le i \le \sigma-3),\quad
g(a_{\sigma-2})=a_2,\quad
g(a_{\sigma-1})=a_1,\quad g(a_\sigma)=a_3.
\]
However, $n$ is not a p-period of $s$. Indeed, if $n \pper[h] s$ for some
bijection $h$, then the block $a_2^n$ is compared with the block $a_3^n$,
forcing $h(a_2)=a_3$. On the other hand, the penultimate block is also
$a_2^n$, and it is compared with the final block $a_1^n$, forcing
$h(a_2)=a_1$. This is a contradiction.

Finally,
\[
|s|=(\sigma+2)n=2n+3n+(\sigma-3)n=p+q+(\sigma-3)d.
\]
Thus, for every $\sigma \ge 2$, we have constructed a string $s$ with
$|\Sigma_s|=\sigma$ and parameterized periods $p$ and $q$ such that
\[
|s|=p+q+(\sigma-3)d
\]
but $d\not\pper s$. This gives a matching lower-bound instance.
\end{proof}

\subsubsection*{Acknowledgements}
This work was supported by JST BOOST Grant Number JPMJBS2406 (RH),
and by JSPS KAKENHI Grant Numbers JP25K00136 (YN), JP23K24808, JP23K18466 (SI).

\bibliographystyle{abbrv}
\bibliography{ref}

\end{document}